\documentclass[a4paper,11pt]{article}

\usepackage[a4paper,margin=1in]{geometry}
\usepackage[UKenglish]{babel}
\usepackage[T1]{fontenc}
\usepackage[utf8]{inputenc}
\usepackage{lmodern}
\usepackage{microtype}
\usepackage{amsmath,amssymb,amsfonts,amsthm}
\usepackage{graphicx}
\usepackage{xcolor}
\usepackage{comment}
\usepackage{url}
\usepackage{hyperref}
\usepackage[nameinlink,noabbrev]{cleveref}
\usepackage{physics}

\input{commands.tex}

\graphicspath{{./figures/}{./}}

\theoremstyle{plain}
\newtheorem{theorem}{Theorem}[section]
\newtheorem{lemma}[theorem]{Lemma}
\newtheorem{corollary}[theorem]{Corollary}

\newcommand{\paperkeywords}{}
\newcommand{\keywords}[1]{\renewcommand{\paperkeywords}{#1}}
\newcommand{\printkeywords}{%
  \par\smallskip\noindent\textbf{Keywords: }\paperkeywords\par
}

\title{Consensus Time in 3-Majority and 2-Choices Is Determined by the Maximum Initial Opinion Density}

\author{Niccol\`o D'Archivio\\
Université Côte d'Azur, Inria, CNRS, i3s, Sophia Antipolis, France\\
\texttt{niccolo.darchivio@inria.fr}\\
\url{https://sites.google.com/view/niccolo-darchivio/home}\\
\href{https://orcid.org/0009-0005-9491-2928}{ORCID: 0009-0005-9491-2928}}

\date{}


\begin{document}

\maketitle

\begin{abstract}
  We establish the correct parameter governing the convergence time of the 3-Majority and 2-Choices dynamics on the complete graph in the synchronous model.
  Recent work [Shimizu and Shiraga, PODC'25] provides matching upper and lower bounds on the number of rounds to consensus, but only in a weak sense: the bounds are shown to coincide for some initial opinion configuration.

  In contrast, we obtain tight bounds in a strong sense, with upper and lower bounds matching up to logarithmic factors for every initial configuration.
  Let $\alpha^{(0)}$ be the initial opinion-frequency vector, and denote by
  $\|\alpha^{(0)}\|_\infty$ its maximum entry.  
  We show that 3-Majority reaches consensus in
  $\tilde{\Theta}(\min\{\|\alpha^{(0)}\|_\infty^{-1},\sqrt n\})$ rounds w.h.p., while 2-Choices reaches consensus in
  $\tilde{\Theta}(\|\alpha^{(0)}\|_\infty^{-1})$ rounds w.h.p.

  Our results demonstrate that the convergence time of both dynamics is governed not by aggregate statistics such as the number of opinions $k$ or the squared $\ell_2$ norm of the initial opinion distribution, but rather by the extremal statistic $\|\alpha^{(0)}\|_\infty$, the maximum initial opinion density.
\end{abstract}

\printkeywords

\section{Introduction}
\label{sec:intro}
We study two well-known opinion dynamics: 3-Majority and 2-Choices. 
These dynamics have received considerable attention in the distributed computing community, as they are simple yet powerful protocols for the consensus problem and its variants \cite{overview2020}.
They are defined on a distributed system of $n$ agents, each holding an opinion in $[k]$, where $2\le k\le n$. The agents interact in synchronous rounds on the complete graph with self-loops.  In each round, every agent samples a few other agents and updates its opinion according to a simple rule.
In 3-Majority, each agent samples three agents and adopts the majority opinion among them, breaking ties uniformly at random.  In 2-Choices, each agent samples two agents and adopts their opinion if they coincide; otherwise, it keeps its current opinion.  The goal is to reach consensus, i.e., a configuration where all agents hold the same opinion.  We are interested in the number of rounds needed to reach consensus with high probability, as a function of the initial opinion configuration.

To describe the state of the system at round $t$, we use the opinion-frequency vector $\alphavec^{(t)}=(\alpha_t(i))_{i\in[k]}$, where $\alpha_t(i)$ is the fraction of agents holding opinion $i$ at round $t$. We refer to $\alphavec^{(t)}$ as the configuration at round $t$.
Two important parameters of a configuration are $\gamma_t=\|\alphavec^{(t)}\|_2^2 = \sum_{i\in[k]} \alpha_t(i)^2$ for the squared $\ell_2$ norm of the opinion distribution, and $m_t=\|\alphavec^{(t)}\|_\infty = \max_{i\in[k]} \alpha_t(i)$ for the maximum opinion density. 
The ratio between these two quantities serves as a proxy for the ``balance'' of a configuration. See \Cref{fig:heatmap} for a visual representation of the ratio $\gamma_t/m_t$ for all possible configurations with three opinions.
While it always holds that $m_t \ge \gamma_t$, equality holds if and only if all nonzero opinion frequencies are equal. Thus, we call $\alphavec^{(t)}$ a \textit{balanced configuration} if $\gamma_t = \Theta(m_t)$. On the other hand, if there are only a few opinions with large density and many with small density, then $m_t$ is much larger than $\gamma_t$, and we call $\alphavec^{(t)}$ an unbalanced configuration.
Recent work by Shimizu and Shiraga~\cite{Shimizu2025} gives upper bounds on the convergence time of 3-Majority and 2-Choices in terms of $\gamma_0$. In particular, they show that 3-Majority reaches consensus in $\tilde{O}(\min\{\gamma_0^{-1},\sqrt n\})$ rounds \whp, while 2-Choices reaches consensus in $\tilde{O}(\gamma_0^{-1})$ rounds \whp.  
They also show that there exists an initial configuration for which these bounds are tight, up to logarithmic factors. Thus, their result leaves open the question of whether $\gamma_0$ is the correct parameter governing the convergence time of these dynamics.

\begin{figure}[htbp]
  \centering
  \def\svgwidth{0.6\columnwidth} 
  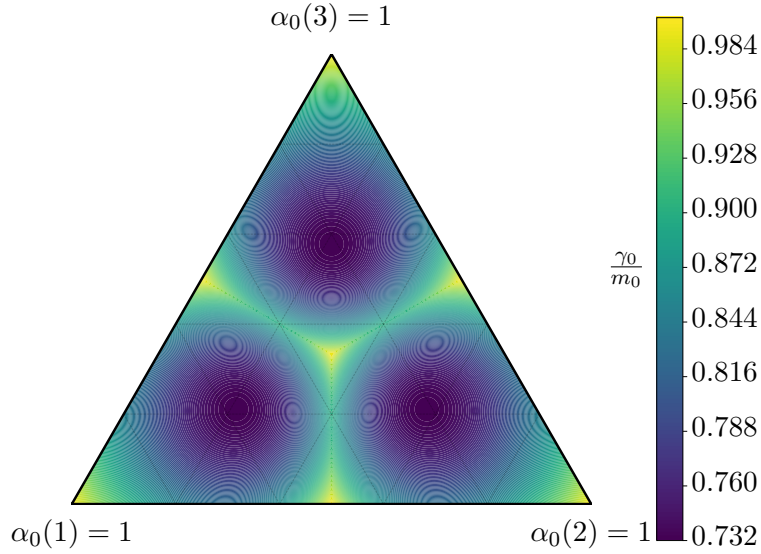
  \caption{A heatmap of the ratio $\gamma_0/m_0$ for all possible initial configurations with three opinions. The corners of the triangle correspond to monochromatic configurations, where $m_0 = \gamma_0 = 1$. The centre of the triangle corresponds to the perfectly balanced configuration, where $m_0 = \gamma_0 = 1/3$. The lower $\gamma_0/m_0$ the more unbalanced the configuration is.}
  \label{fig:heatmap}
\end{figure}

\subsection{Our Contribution}
We show that the correct parameter governing the convergence time of both dynamics is $m_0$, the maximum initial opinion density, rather than $\gamma_0$.  
In particular, we show that 3-Majority reaches consensus in $\tilde{\Theta}(\min\{m_0^{-1},\sqrt n\})$ rounds \whp, while 2-Choices reaches consensus in $\tilde{\Theta}(m_0^{-1})$ rounds \whp. More precisely, we prove the following result.

\begin{theorem}[Main result, informal]
\label{thm:main-informal}
  Consider a configuration with opinion-frequency vector $\alphavec^{(0)}$, and let $m_0=\|\alphavec^{(0)}\|_\infty$ be the maximum initial opinion density.
  If the network is the complete graph with self-loops and agents activate synchronously, then 3-Majority reaches
  consensus in
  \(
    \tilde{\Theta}\!\left(\min\{m_0^{-1},\sqrt n\}\right)
  \)
  rounds with high probability, while 2-Choices reaches consensus in
  $\tilde{\Theta}(m_0^{-1})$ rounds with high probability. See \Cref{thm:main,thm:main-lower} for the formal statement with the precise logarithmic factors.
\end{theorem}

This result translates into a substantial improvement over the state-of-the-art upper bounds of Shimizu and Shiraga~\cite{Shimizu2025} for unbalanced initial configurations, where $m_0$ is much larger than $\gamma_0$.  For example, consider a configuration with $n$ opinions: one opinion has density $m_0=1/\sqrt{n}$ and the remaining $n-1$ opinions each have density $(1-m_0)/(n-1)$, so $\gamma_0 = m_0^2 + (1-m_0)^2/(n-1) \approx 1/n$. In this configuration, the convergence time of 2-Choices is $\tilde{\Theta}(\sqrt{n})$ by our result, while the bound of Shimizu and Shiraga would give only $\tilde{O}(n)$.
A similar gap appears for 3-Majority, leading to improvements up to a factor $n^{1/4}$ on consensus time.

Beyond convergence time improvements, our result sheds light on an interesting fact: the consensus time of these two dynamics does not depend on aggregate statistics such as the number of opinions $k$ or the squared $\ell_2$ norm $\gamma_0$ of the initial opinion distribution, but rather on the extremal statistic $m_0$, the maximum initial opinion density.
This is somewhat counterintuitive: if we fix the largest opinion density $m_0$, then the distribution of the remaining mass among the other opinions (which changes $\gamma_0$) does not affect the convergence time of the dynamics. At first glance, one could expect that distributing the remaining mass among a few competing opinions would make the convergence time longer, while distributing the same mass among more opinions would make it shorter, but this is not the case.

Another interesting consequence of our result is that we show for the first time a gap in the convergence time between 3-Majority and the Undecided-State dynamics.
The Undecided-State is a well-known opinion dynamics in which an agent that samples an opinion different from its own becomes undecided, and an undecided agent just copies the sampled opinion.
Becchetti et al.~\cite{Becchetti2016} proved that the Undecided-State reaches consensus in $O(\frac{\gamma_0}{m_0^2}\log n)$ rounds \whp, as long as the number of opinions $k=O((n\log n)^{1/3})$.
This means that 3-Majority can be slower than the Undecided-State by a factor of roughly $m_0/\gamma_0$, which becomes substantial for unbalanced configurations.
To the best of our knowledge, this surprising advantage of the Undecided-State over 3-Majority has passed unnoticed until now, as it exists only for unbalanced configurations, highlighting the importance of analysing the dynamics for every initial configuration. See \Cref{fig:comparison_conv_time} for a comparison of the convergence time of 3-Majority, 2-Choices, and the Undecided-State for different initial configurations.

An important variant of the consensus problem is the plurality consensus problem, where the goal is to reach consensus on the initial plurality opinion, i.e., the opinion with the largest initial density. We show that, for both 3-Majority and 2-Choices, plurality consensus is ensured by conditions on the extremal densities, namely $m_0$ and the bias between the largest and second-largest opinion densities. More precisely, we prove the following result.

\begin{theorem}[Plurality consensus, informal]
\label{thm:plurality-informal}
  Consider a configuration with opinion-frequency vector $\alphavec^{(0)}$ such that $\alpha_0(1) > \alpha_0(2) \ge \cdots \ge \alpha_0(k)$. Let $m_0=\alpha_0(1)$ be the maximum initial opinion density, and let $\bias_0 = \alpha_0(1) - \alpha_0(2)$ be the density bias between the largest and second-largest opinions.
  If the network is the complete graph with self-loops and agents activate synchronously, then under the conditions $m_0 = \omega(\frac{\log n}{\sqrt{n}})$ and $\bias_0 = \omega(\sqrt{\frac{\log n}{n}})$, 3-Majority reaches consensus on opinion 1 in $O(m_0^{-1} \log n)$ rounds with probability at least $1-n^{-2}$, while for 2-Choices, the same statement holds under the condition $m_0 = \omega(\frac{(\log n)^2}{n})$ and $\bias_0 = \omega(\sqrt{\frac{m_0 \log n}{n}})$.
\end{theorem}

This result relaxes the state-of-the-art conditions for plurality consensus of Shimizu and Shiraga~\cite{Shimizu2025}, which require $\gamma_0 = \omega(\frac{\log n}{\sqrt{n}})$ for 3-Majority, and $\gamma_0 = \omega(\frac{(\log n)^2}{n})$ for 2-Choices. Again, the improvement is substantial for unbalanced configurations, and the gap between our conditions and those of Shimizu and Shiraga can be as large as a factor of $n^{1/4}$ for 3-Majority and $n^{1/2}$ for 2-Choices.

\begin{figure}[htbp]
  \centering
  \def\svgwidth{0.9\columnwidth} 
  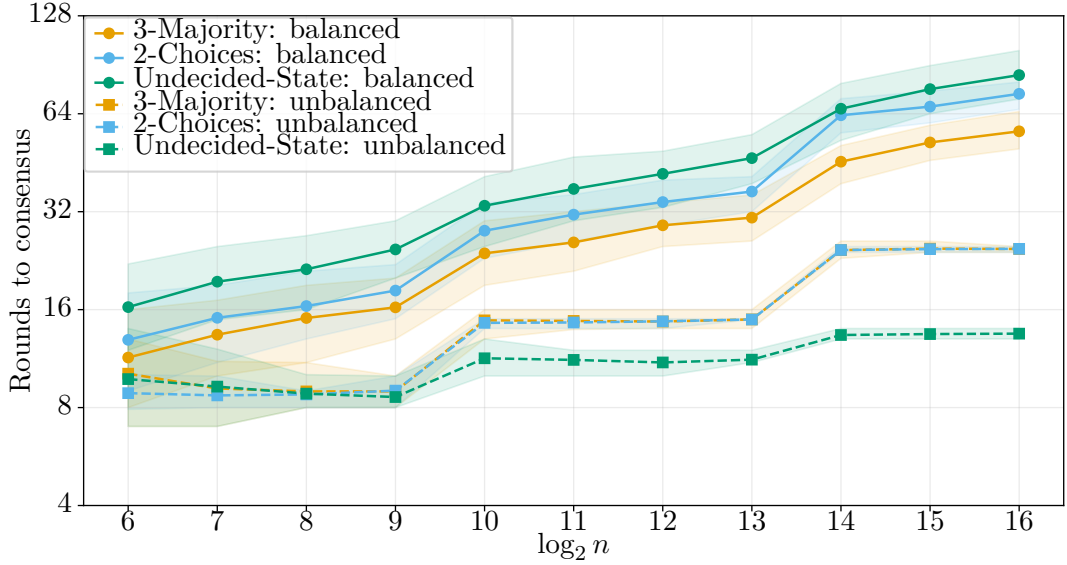
  \caption{Comparison of consensus time between 3-Majority, 2-Choices, and Undecided-State over 100 simulations. We fix the maximum opinion density $m_0 \approx n^{-1/4}$, and we compare two different distributions of the remaining mass among the other opinions. In the balanced case, all opinions have mass $m_0$, while in the unbalanced case we set all the remaining opinion density to $1/n$.}
  \label{fig:comparison_conv_time}
\end{figure}


\section{Related Work}\label{sec:related_work}
There is a vast body of research investigating the convergence time of simple dynamics such as \threemaj, \twochoices, and the \undecided, both in the synchronous and in the asynchronous settings.
We first focus on the works on the synchronous setting.

\textbf{3-Majority.}
Let us specify that all the results we mention in this subsection and the next two subsections hold on the complete graph of \(n\) nodes with self-loops.
The work \cite{BecchettiCNPT16} proved an upper bound of \(O\pa{(k^2 \sqrt{\log n} + k \log n)(k + \log n)}\) rounds to reach consensus that holds w.h.p., provided that \(k \le n^{\alpha}\) for a suitable positive constant \(\alpha < 1\).

In \cite{becchettiSimpleDynamicsPlurality2017}, the authors showed that the synchronous \threemaj dynamics with \(k\) opinions converges in \(O(\min\{k, (n/\log n)^{\frac 13}\} \log n)\) with high probability, provided that the bias of the initial configuration is at least \(c \sqrt{\min\{2k, (n/\log n)^{\frac 13}\} n \log n}\) for some constant \(c > 0\).
Moreover, the authors provided a lower bound of \(\Omega(k\log n)\) on the convergence time to consensus, w.h.p., whenever the initial configuration is sufficiently balanced, that is, \(\max_{i \in [k]}\{\conf_0(i)\} \le n/k + (n/k)^{1 - \varepsilon}\) for some \(\varepsilon > 0\) and \(k \le (n/\log n)^{1/4}\).

\textbf{2-Choices.}
The work \cite{BerenbrinkCEKMN17} compared the synchronous \threemaj dynamics with the synchronous \twochoices dynamics.
The authors first proved a generic lower bound of \(\Omega(\min\{k, n/\log n\})\) rounds to reach consensus starting from the initial perfectly balanced configuration, w.h.p.
Furthermore, they proved that the \threemaj dynamics works better in symmetric configurations (i.e., with no initial bias) when, e.g., \(\max_{i \in [k]}\{\conf_0(i)\} = O(\log n)\).
In particular, the \threemaj takes time at most \(O(n^{3/4} \log^{7/8} n)\) to reach consensus w.h.p., regardless of any other hypothesis on the initial configuration, while the \twochoices needs time \(\Omega(n / \log n)\) whenever \(\max_{i \in [k]}\{\conf_0(i)\} = O(\log n)\).
This was the first work to notice that, for a large number of opinions, the \threemaj dynamics is polynomially (in \(k\)) faster than the \twochoices dynamics.

The work \cite{Ghaffari2018} improved upon \cite{becchettiSimpleDynamicsPlurality2017} and showed that, for the \twochoices with \(k = O(\sqrt{n / \log n})\) and for the \threemaj with \(k = O(n^{1/3}/\sqrt{\log n})\), the convergence time to consensus is \(O(k \log n)\), with high probability.
Notice that this upper bound is tight according to the lower bound by \cite{becchettiSimpleDynamicsPlurality2017}, at least as long as \(k \le (n/\log n)^{1/4}\).
Furthermore, the authors showed that the convergence time of the \threemaj dynamics is \(O(n^{2/3} \log^{3/2} n)\) with high probability, regardless of the number of opinions.

Very recently, \cite{Shimizu2025} almost tightly settled the complexity of both the \threemaj and the \twochoices dynamics.
The authors proved that, w.h.p., the \threemaj dynamics reaches consensus in \(O(k\log n)\) rounds whenever \(k = o(\sqrt{n}/\log n)\), while it takes time \(O(\sqrt{n} \log^2 n )\) for other values of \(k\).
Furthermore, \cite{Shimizu2025} proved that plurality consensus is ensured w.h.p.\ as long as the initial count gap is \(\omega(\sqrt{n \log n})\).
As for the \twochoices, they showed that, w.h.p., the dynamics reaches consensus in \(O(k\log n)\) rounds whenever \(k = o(n/\log^2 n)\), while it takes time \(O(n \log^3 n)\) otherwise.
In density notation, these plurality requirements are \(\bias_0=\omega(\sqrt{\log n/n})\) for \threemaj and \(\bias_0 = \omega(\sqrt{\alpha_0(1)\log n/n})\) for \twochoices.\footnote{This is the only example of initial bias that gets ``close enough'' to what we require in \Cref{thm:plurality-informal}, but with an extra \(\sqrt{\log n}\) factor.}
These results almost match the generic lower bound given by \cite{BerenbrinkCEKMN17}, up to logarithmic factors.

\textbf{Undecided-State Dynamics.}
We already introduced another well-known consensus dynamics, the \undecided.
In the \undecided there is an extra opinion, the 
\textit{undecided} opinion.
The update rule works as follows:
A node samples one neighbour u.a.r.\ and pulls its opinion.
If the received opinion is different from the one it currently supports, the node becomes undecided.
If the node is undecided, it just copies whatever opinion it receives.
The \undecided has been studied both in the synchronous setting and in the population protocol model (asynchronous setting) \cite{AngluinAE08,ClementiGGNPS18,Becchetti2016,DAmoreCN20,DAmoreCN22,AmirABBHKL23,berenbrink2024,BankhamerBBEHKK22} and has often been considered to perform roughly the same as the \threemaj dynamics.
A special mention goes to \cite{Becchetti2016}, which analysed the \undecided dynamics in the synchronous setting and proved that it reaches consensus in $O(\frac{\gamma_0}{m_0^2} \log n)$ rounds \whp, as long as the number of opinions $k=O((n\log n)^{1/3})$. They also proved a matching lower bound up to logarithmic factors as long as $k=O((n\log n)^{1/6})$.
Prior to this work, this was the only result on the \undecided dynamics that provided strong and tight bounds, i.e., bounds that hold for every initial configuration.
A recent result of Cooper et al. \cite{cooper2026undecidedstatedynamicsopinions} shows that the Undecided-State reaches consensus in $O(\min\{ k, \sqrt{n} \}\log n)$ rounds \whp, regardless of the number of opinions.

\subsubsection*{Asynchronous setting.}
The only works that analysed the \threemaj dynamics in the asynchronous setting are \cite{CooperMRSS25,BerenbrinkCGHKR22}.
In \cite{BerenbrinkCGHKR22}, the authors consider the binary opinion case and show that the convergence time of the asynchronous \threemaj dynamics is \(O(n \log n)\) rounds, w.h.p., and that a bias of \(\Theta(\sqrt{n \log n})\) is sufficient to ensure plurality consensus, w.h.p.
The authors of \cite{CooperMRSS25} showed that the convergence time is \(O(\min\{kn\log^2 n, n^{3/2} \log^{3/2} n\})\), w.h.p., regardless of the number of initial opinions. 
They also provided a generic lower bound of \(\Omega(\min\{kn, n^{3/2}/\sqrt{\log n}\})\) rounds to reach consensus (starting from balanced configurations), w.h.p.
The work \cite{CooperMRSS25} (which came before \cite{Shimizu2025}) was the first to establish exactly how the linear-in-\(k\) dependence in the consensus time of the \threemaj dynamics breaks when the number of opinions exceeds \(\sqrt{n}\), in which case the consensus time is sublinear in the number of opinions.
The reader may observe that the asynchronous setting has the same qualitative convergence time as the synchronous model, up to a multiplicative factor \(n\).
This is to be expected: 
In the asynchronous setting, in a round, only one agent (sampled u.a.r.) updates its state, and hence we need roughly \(n\) rounds to activate all agents at least once (up to polylogarithmic factors).
For processes with small variance (smaller than that of the voter model), convergence times from the synchronous to the asynchronous setting usually scale with such a factor.
Note that this is not true for processes with large variance \cite{BecchettiCPTVZ24}.

For additional related work, we refer the reader to \Cref{sec:beyond_constant_samples}.

\section{Proof Overview}\label{sec:proof-overview}
The proof idea behind the upper bound of \Cref{thm:main-informal} is simple. For both 3-Majority and 2-Choices, it is possible to compute the expectation of $\alpha_{t+1}(i)$ conditional on the configuration $\alpha^{(t)}$ at the previous round, and to show that
\[
  \E[\alpha_{t+1}(i) \mid \alpha^{(t)}]
  =
  \alpha_{t}(i)\pa{ 1 + \alpha_t(i) - \gamma_t }.
\]
As long as we have an unbalanced configuration, say $m_t \geq 2 \gamma_t$, and $i^\star$ is an opinion with density $m_t$, then $\E[\alpha_{t+1}(i^\star) \mid \alpha^{(t)}] = m_t \pa{ 1 + m_t - \gamma_t } \geq m_t (1+ \frac{1}{2}m_t)$.
Assuming that $m_t = \Omega(m_0)$, the density of such an opinion grows at every round by a multiplicative factor of $(1+ \frac{1}{2}m_0)$ in expectation. Thus, if the configuration remains unbalanced up to consensus, then the convergence time is $O(m_0^{-1} \log n)$, which is the desired upper bound.
Otherwise, let $\tau_\text{b}$ be the first round in which the configuration becomes balanced, i.e., $\gamma_{\tau_\text{b}} \geq m_{\tau_\text{b}}/2$. $\tau_\text{b} = O(m_0^{-1} \log n)$ since the monochromatic configuration is balanced.
We can apply Shimizu and Shiraga's result to the configuration at round $\tau_\text{b}$, which is balanced, to obtain consensus in $\tilde{O}(\gamma_{\tau_\text{b}}^{-1}) = \tilde{O}(m_0^{-1})$ additional rounds.
To make this argument rigorous, we concentrate around the expected growth of $m_t$ in the unbalanced phase by using the Freedman inequality for martingales as in \cite{Shimizu2025}.

For the lower bound of \Cref{thm:main-informal}, we partition the opinion
labels into blocks of comparable initial mass and track the aggregate mass of
each block.  For a fixed block $S$, let
$x_t=\sum_{i\in S}\alpha_t(i)$ and
$q_t=\sum_{i\in S}\alpha_t(i)^2$.  For both dynamics,
$\E_t[x_{t+1}]=x_t+q_t-x_t\gamma_t$, so the positive drift is at most
$q_t\le x_t^2$.  A stopped martingale
argument shows that a block of initial mass $x_0$ grows by a constant factor
within $O(x_0^{-1})$ rounds with probability at most
$\exp(-\Omega(nx_0^2))$ for 3-Majority and $\exp(-\Omega(nx_0))$ for
2-Choices.  In the latter case, we obtain the sharper exponent by counting
only vertices that enter the block.  Consensus forces the block containing
the winning opinion to undergo such an upcrossing, and a union bound over the
blocks gives the claimed result.

We also extend our analysis to the plurality consensus problem. The idea is very similar to the one for consensus time. We know that the bias between the largest and any remaining opinion densities grows at every round by a multiplicative factor of $(1+ \frac{1}{2}m_0)$ in expectation, as long as the configuration is unbalanced. Thus, with the same Freedman inequality argument, we prove that the condition on the bias in \Cref{thm:plurality-informal} is preserved until the configuration becomes balanced, and then we apply Shimizu and Shiraga's result to the balanced configuration to obtain plurality consensus. 
The argument is slightly more technical for 2-Choices, as the bias growth depends on the scaled bias $\eta_t(i,j)= \frac{\delta_t(i,j)}{\sqrt{\max\{\alpha_t(i),\alpha_t(j)\}}}$ rather than the bias $\delta_t(i,j)$, and we need to show that the condition on $\eta_0(i,j)$ in \Cref{thm:plurality-informal} is preserved until the configuration becomes balanced. To do this, we show that the bias grows at a higher rate than $\max\{\alpha_t(i),\alpha_t(j)\}$. In fact $\delta_t(i,j)$ grows at every round by a multiplicative factor of $(1+ \alpha_t(i) + \alpha_t(j) - \gamma_t)$ in expectation, while $\max\{\alpha_t(i),\alpha_t(j)\}$ grows by a multiplicative factor of $(1+ \max\{\alpha_t(i),\alpha_t(j)\} - \gamma_t)$ in expectation. 
\section{Model and Notation}
\label{sec:model}

We consider the complete graph on $V=[n]$ with self-loops.  Each vertex holds
an opinion in $\opinions=[k]$, where $2\le k\le n$.  At round $t$, let
$\alpha_t(i)$ denote the fraction of vertices holding opinion $i$, and write
$\alphavec^{(t)}=(\alpha_t(i))_{i\in\opinions}$.  We use
\[
  m_t=\|\alphavec^{(t)}\|_\infty=\max_{i\in\opinions}\alpha_t(i),
  \qquad
  \gamma_t=\|\alphavec^{(t)}\|_2^2=\sum_{i\in\opinions}\alpha_t(i)^2 .
\]
For two opinions $i,j$, define the signed bias
$\delta_t(i,j)=\alpha_t(i)-\alpha_t(j)$.  For 2-Choices we also use the scaled
bias
\[
  \eta_t(i,j)=
  \frac{\delta_t(i,j)}
       {\sqrt{\max\{\alpha_t(i),\alpha_t(j)\}}}.
\]

\textbf{3-Majority dynamics.}
In each round every vertex samples three vertices independently and uniformly
from $V$.  It adopts the majority opinion among the three samples; if all three
sampled opinions are different, the tie is broken uniformly among them.

\textbf{2-Choices dynamics.}
In each round every vertex samples two vertices independently and uniformly
from $V$.  If the two sampled vertices have the same opinion, the vertex adopts
that opinion; otherwise it keeps its current opinion.

\textbf{Notation.}
We denote the consensus time by
\[
  \tau_{\mathrm{cons}}
  =
  \inf\{t\ge0:\exists i\in\opinions \text{ such that } \alpha_t(i)=1\},
\]
and the first balancing time by
\[
  \tau_{\mathrm b}
  =
  \inf\{t\ge0:\gamma_t\ge m_t/2\}.
\]
We use the convention $\inf\emptyset=\infty$.
We say an event happens \emph{with high probability} (w.h.p.) if it happens with probability at least $1-n^{-2}$.


\section{Analysis: Consensus time}\label{sec:quick-proof}
In this section we state and prove a formal version of \Cref{thm:main-informal}.
\subsection{Upper bound.}
\begin{theorem}[Consensus-time upper bound]\label{thm:main}
Consider any initial configuration on the complete graph with self-loops, and
let $m_0=\|\alphavec^{(0)}\|_\infty$.  There is a universal constant $C>0$ such
that the following bounds hold with high probability.  For 3-Majority,
\[
  \tau_{\mathrm{cons}}
  \le
  \begin{cases}
    C m_0^{-1}\log n, & \text{if } m_0\ge C\log n/\sqrt n,\\
    C\sqrt n\,\log^2 n, & \text{otherwise}.
  \end{cases}
\]
For 2-Choices,
\[
  \tau_{\mathrm{cons}}
  \le
  \begin{cases}
    C m_0^{-1}\log n, & \text{if } n m_0\ge C(\log n)^2,\\
    Cn\log^3 n, & \text{otherwise}.
  \end{cases}
\]
In particular, 3-Majority reaches consensus in
$\tilde O(\min\{m_0^{-1},\sqrt n\})$ rounds and 2-Choices reaches consensus in
$\tilde O(m_0^{-1})$ rounds, as stated in \Cref{thm:main-informal}.
\end{theorem}

In the following lemma, we show that from an unbalanced configuration, the process reaches a balanced configuration in $O(m_0^{-1})$ rounds. We also show that, up to that time, $m_t = \Omega(m_0)$.

\begin{lemma}[Fast balancing from the infinity norm]
\label{lem:quick-balance}
Recall the balancing time $\tau_{\mathrm b}$ from \Cref{sec:model}.
There are universal constants $C,c>0$ such that the following
holds.  For
3-Majority, if $n m_0^2\ge C\log n$, then with probability at least
$1-n^{-5}$ there is a time $s\le C m_0^{-1}$ such that
\[
  s=\tau_{\mathrm b},
  \qquad
  m_t\ge c\, m_0
  \quad\text{for every }0\le t\le s .
\]
For 2-Choices, the same statement holds under the condition
$n m_0\ge C(\log n)^2$.
\end{lemma}

\begin{proof}
If $\gamma_0\ge m_0/2$, there is nothing to prove, taking $s=0$. Otherwise we run the
process in dyadic phases.  Fix a phase scale $x\ge m_0$ and a phase-starting
round $r$ with $m_r\ge x$ and $\gamma_r<m_r/2$.  Let
\[
  \sigma=\inf\{u\ge r:\gamma_u\ge m_u/2
  \text{ or } m_u\ge 2x \text{ or } m_u<x/2\}
\]
and let $H=\lceil 20/x\rceil$.  For $u<\sigma$, choose a plurality opinion
$i_u$ by deterministic tie-breaking, so that $\alpha_u(i_u)=m_u$.  Since
$i_u$ is $\calF_u$-measurable, \Cref{item:alpha} of
\Cref{lem:basic inequality} gives
\[
  \E_u[\alpha_{u+1}(i_u)]
  =m_u(1+m_u-\gamma_u)
  \ge m_u(1+m_u/2).
\]
Thus, while $u<\sigma$,
\[
  \E_u[\alpha_{u+1}(i_u)]-m_u\ge m_u^2/2\ge x^2/8.
\]
Moreover, $m_{u+1}\ge\alpha_{u+1}(i_u)$.

Define the stopped martingale, for $0\le q\le H$,
\[
  M_q=\sum_{u=r}^{r+q-1}\indicator_{\{u<\sigma\}}
  \bigl(\alpha_{u+1}(i_u)-\E_u[\alpha_{u+1}(i_u)]\bigr).
\]
By \Cref{item:BC for alpha} of
\Cref{lem:Bernstein condition for sync processes}, the lower-tail increments
$\E_u[\alpha_{u+1}(i_u)]-\alpha_{u+1}(i_u)$ satisfy a one-sided Bernstein
condition with $\bounded=1/n$.  During the phase, $m_u<2x$ and
$\gamma_u<m_u/2<x$, so the variance parameter is at most
\[
  \variance_x=
  \begin{cases}
    2x/n, & \text{for 3-Majority},\\
    6x^2/n, & \text{for 2-Choices}.
  \end{cases}
\]
Applying \Cref{cor:Freedman} to $-M_q$ with $h=x/2$ yields, for universal
constants $c_3,c_2>0$,
\[
  \Prob\!\left[\exists q\le H:\ M_q\le -x/2\mid\calF_r\right]
  \le
  \begin{cases}
    \exp(-c_3 n x^2), & \text{for 3-Majority},\\
    \exp(-c_2 n x), & \text{for 2-Choices}.
  \end{cases}
\]
To make the implication for the phase explicit, for every $q\le H$ with
$r+q\le\sigma$, we have
\begin{align*}
  m_{r+q}-m_r
  &\ge \sum_{u=r}^{r+q-1}\bigl(\alpha_{u+1}(i_u)-m_u\bigr)\\
  &=M_q+\sum_{u=r}^{r+q-1}
    \bigl(\E_u[\alpha_{u+1}(i_u)]-m_u\bigr)\\
  &\ge M_q+q x^2/8.
\end{align*}
At a first downcrossing at time $r+q$, we would have $m_r\ge x$ and
$m_{r+q}<x/2$, so the last display would imply
\[
  M_q+q x^2/8\le m_{r+q}-m_r<-x/2,
\]
and therefore $M_q<-x/2$.  Thus, on the complement of the displayed
Freedman event, the phase cannot end through the condition $m_u<x/2$.
If the phase has not ended by time $r+H$, the same inequality and
$H=\lceil20/x\rceil$ give
\[
  m_{r+H}
  \ge m_r+H x^2/8+M_H
  >x-x/2+20x/8
  =3x>2x,
\]
again a contradiction.  Therefore, except with the displayed probability, the
phase ends within $H=O(1/x)$ rounds either by hitting $\tau_{\mathrm b}$ with
$m_{\tau_{\mathrm b}}\ge x/2$, or by reaching mass at least $2x$.

Starting with $x=m_0$ and iterating over $x,2x,4x,\ldots$ until $x\ge1/4$,
the total number of rounds is
\[
  \sum_q O((2^q m_0)^{-1})=O(m_0^{-1}).
\]
If no earlier balanced time is reached, then the final doubling gives
$m_t\ge1/2$, which itself implies $\gamma_t\ge m_t^2\ge m_t/2$.  Thus
$t=\tau_{\mathrm b}$ and $m_t=\Omega(m_0)$.
Moreover, on the same successful event, no phase ends through the stopping
condition $m_u<x/2$.  Hence, throughout each phase and until the balancing
time is reached, $m_u\ge x/2\ge m_0/2$.  This gives the stated lower bound on
$m_t$ for every $t\le\tau_{\mathrm b}$, after setting
$c=1/2$.
There are at most $1+\lceil\log_2 n\rceil=O(\log n)$ phases, since
$m_0\ge1/n$.  Moreover, every phase has $x\ge m_0$.  Under
$n m_0^2\ge C\log n$ for 3-Majority, choosing the universal constant $C$
large enough makes each phase failure probability at most $n^{-7}$.
Under $n m_0\ge C(\log n)^2$ for 2-Choices, increasing $C$ if necessary
gives the same bound (in fact, a stronger one).  A union bound over the phases
is therefore at most $n^{-5}$.
\end{proof}

Now, we can complete the proof of \Cref{thm:main} by applying the
large-\(\gamma\) consensus bound in Theorem~2.1 of Shimizu and
Shiraga~\cite{Shimizu2025}
from the random time $\tau_{\mathrm b}$, conditioned on the current
configuration.

\begin{proof}[Proof of \Cref{thm:main}]
Consider first either dynamics in the regime where its high-mass condition
holds: $m_0\ge C\log n/\sqrt n$ for 3-Majority, or
$n m_0\ge C(\log n)^2$ for 2-Choices.  We condition on the event of
\Cref{lem:quick-balance}.  At the balancing time $s$ we have
\[
  \gamma_s\ge m_s/2=\Omega(m_0).
\]
In the 3-Majority case, by increasing $C$ if necessary this gives
$\gamma_s\ge C_{\mathrm{SS}}\log n/\sqrt n$, where $C_{\mathrm{SS}}$ is the
constant required by Shimizu and Shiraga's Theorem~2.1.  Applying that theorem
from the random time $s$, conditioned on the current
configuration, gives consensus in
\[
  O\!\left(\gamma_s^{-1}\log n\right)
  \le
  O\!\left(m_0^{-1}\log n\right)
\]
additional rounds.  For 2-Choices the same black-box theorem gives
\[
  O(\gamma_s^{-1}\log n)=O(m_0^{-1}\log n)
\]
additional rounds in the high-mass regime $n m_0\ge C(\log n)^2$.

If the 3-Majority high-mass condition fails, then
$m_0<C\log n/\sqrt n$, and the original Shimizu--Shiraga all-configuration
bounds in Theorems~2.1 and~2.2 of Shimizu and Shiraga~\cite{Shimizu2025}
give the stated $O(\sqrt n\log^2 n)$ branch.  For 2-Choices below
$n m_0\ge C(\log n)^2$, we use the corresponding norm-growth fallback from
the same two theorems, incurring their stated polylogarithmic loss.

\end{proof}

\subsection{Lower bound}\label{sec:quick-lower-bound}
In this section we state and prove a formal version of the lower bound in \Cref{thm:main-informal}.
\begin{theorem}[Consensus-time lower bound]\label{thm:main-lower}
Assume that the initial configuration is not already at consensus.  There is a
universal constant $c>0$ such that, with probability at least $1-n^{-3}$,
3-Majority satisfies
\[
  \tau_{\mathrm{cons}}
  \ge
  c\min\!\left\{m_0^{-1},\sqrt{\frac{n}{\log n}}\right\},
\]
and 2-Choices satisfies
\[
  \tau_{\mathrm{cons}}
  \ge
  c\min\!\left\{m_0^{-1},\frac{n}{\log n}\right\}.
\]
\end{theorem}

The next lemma extends the single-opinion upcrossing estimate to the aggregate
mass of any fixed set of opinions.

\begin{lemma}[Aggregate-mass upcrossing]
\label{lem:aggregate-upcrossing}
Fix a nonempty set $S\subseteq\opinions$ and write
\[
  x_t=\sum_{i\in S}\alpha_t(i),
  \qquad
  q_t=\sum_{i\in S}\alpha_t(i)^2.
\]
Assume that $0<x_0\le3/4$.  There are universal constants
$c,c_3,c_2>0$ such that, for $\varepsilon=1/10$,
\[
  \Pr\!\left[\exists t\le \frac{c}{x_0}:
    x_t\ge(1+\varepsilon)x_0\right]
  \le
  \begin{cases}
    \exp(-c_3n x_0^2), & \text{for 3-Majority},\\
    \exp(-c_2n x_0), & \text{for 2-Choices}.
  \end{cases}
\]
\end{lemma}

\begin{proof}
Summing \Cref{item:alpha} of \Cref{lem:basic inequality} over $i\in S$
gives, for both dynamics,
\begin{equation}
\label{eq:aggregate-mass-expectation}
  p_t:=\E_t[x_{t+1}]
  =x_t+q_t-x_t\gamma_t.
\end{equation}
Thus $p_t-x_t\le q_t\le x_t^2$.  Let
$\tau=\inf\{t\ge0:x_t\ge(1+\varepsilon)x_0\}$ and
$H=\lfloor c/x_0\rfloor$.
Since round indices are integers, the event in the lemma is exactly
$\{\tau\le H\}$.
For $t<\tau$, both $x_t$ and $q_t$ are at most a universal constant times
$x_0$ and $x_0^2$, respectively.  Below, the universal constant $c$ is chosen
sufficiently small compared with $\varepsilon$.

\textbf{3-Majority.}
Conditioned on $\calF_t$, the updates of the vertices are independent and
$n x_{t+1}\sim\operatorname{Bin}(n,p_t)$.  Consequently, the centered increment
$x_{t+1}-p_t$ satisfies a one-sided Bernstein condition with boundedness
parameter $1/n$ and variance proxy at most $p_t/n\le Cx_0/n$ for
$t<\tau$, where $C>0$ is universal.  Define the stopped martingale
\[
  M_s=\sum_{t=0}^{s-1}\indicator_{\{t<\tau\}}(x_{t+1}-p_t),
  \qquad 0\le s\le H.
\]
Its predictable variance up to time $H$ is at most $Cc/n$.  If
$\tau\le H = \left\lfloor\frac{c}{x_0}\right\rfloor$, then \eqref{eq:aggregate-mass-expectation} and the fact that $p_t-x_t\le x_t^2\le (1+\varepsilon)^2 x_0^2
   \le Cx_0^2$
give
\[
  M_\tau
  =x_\tau-x_0-\sum_{t=0}^{\tau-1}(p_t-x_t)
  \ge \varepsilon x_0-Cc x_0
  \ge \frac{\varepsilon x_0}{2}.
\]
By \Cref{cor:Freedman-predictable},
$\Pr[\tau\le H]\le \exp(-c_3n x_0^2)$ for a universal constant $c_3>0$.

\textbf{2-Choices.}
Let $G_{t+1}$ and $L_{t+1}$ be, respectively, the fractions of vertices that
enter and leave the set of opinions $S$ in round $t+1$.  Then
$x_{t+1}-x_t=G_{t+1}-L_{t+1}\le G_{t+1}$.
An outside vertex enters $S$ exactly when its two samples have the same
opinion in $S$.  Hence, conditioned on $\calF_t$,
$nG_{t+1}\sim\operatorname{Bin}(n(1-x_t),q_t)$.  Put
$g_t=\E_t[G_{t+1}]=(1-x_t)q_t$.  For $t<\tau$,
$g_t\le Cx_0^2$ and $\Var_t(G_{t+1})\le Cx_0^2/n$.
More precisely, $G_{t+1}-g_t$ satisfies a one-sided Bernstein condition with
boundedness parameter $1/n$ and variance proxy $g_t/n$.
The stopped process
\[
  N_s=\sum_{t=0}^{s-1}\indicator_{\{t<\tau\}}(G_{t+1}-g_t),
  \qquad 0\le s\le H,
\]
is a martingale whose predictable variance is at most $Cc x_0/n$.  If
$\tau\le H$, then $\sum_{t<\tau}G_{t+1}\ge x_\tau-x_0\ge\varepsilon x_0$,
whereas $\sum_{t<\tau}g_t\le Cc x_0$.  Therefore
$N_\tau\ge\varepsilon x_0/2$.  Another application of
\Cref{cor:Freedman-predictable} gives
$\Pr[\tau\le H]\le \exp(-c_2n x_0)$ for a universal constant $c_2>0$.
\end{proof}

We now prove the lower bound by grouping the initial opinions into blocks of
comparable total mass and applying \Cref{lem:aggregate-upcrossing} to every
block.

\begin{proof}[Proof of \Cref{thm:main-lower}]
Fix a sufficiently large universal constant $A>0$ and set
\[
  a_\star=
  \begin{cases}
    A\sqrt{\dfrac{\log n}{n}}, & \text{for 3-Majority},\\[2mm]
    A\dfrac{\log n}{n}, & \text{for 2-Choices},
  \end{cases}
  \qquad
  a=\max\{m_0,a_\star\}.
\]
If $a>1/4$, the quantity in the corresponding minimum in
\Cref{thm:main-lower} is at most $4A$.  Thus the claimed lower bound follows
from $\tau_{\mathrm{cons}}\ge1$ after decreasing the universal constant in the
theorem.  Hence assume $a\le1/4$.

Ignore opinions of initial mass zero, which can never appear.  Since every
remaining opinion has initial mass at most $m_0\le a$, greedily partition the
opinion labels into blocks $S_1,\ldots,S_L$ as follows.  Add opinions to the
current block until its initial mass first reaches $a$, then start a new block.
Append the final remainder, if nonempty, whose mass is less than $a$, to the
preceding block.  Writing $b_r=\sum_{i\in S_r}\alpha_0(i)$, we have
\[
  a\le b_r<3a
  \qquad\text{for every }r,
  \qquad
  L\le\frac1a\le n.
\]

Let $c$ be the time constant in \Cref{lem:aggregate-upcrossing} and put
$T=c/(3a)$.  Since $b_r<3a$, we have $T<c/b_r$.  Moreover,
$(1+\varepsilon)b_r<1$, because $b_r<3a\le3/4$ and
$\varepsilon=1/10$.  Thus, if consensus is reached by time $T$, the block
containing the winning opinion must first increase from mass $b_r$ to at least
$(1+\varepsilon)b_r$.  By \Cref{lem:aggregate-upcrossing} and a union bound,
\[
  \Pr[\tau_{\mathrm{cons}}\le T]
  \le
  \begin{cases}
    L\exp(-c_3n a^2), & \text{for 3-Majority},\\
    L\exp(-c_2n a), & \text{for 2-Choices}.
  \end{cases}
\]
For 3-Majority, $na^2\ge A^2\log n$, and for 2-Choices,
$na\ge A\log n$.  Taking $A$ sufficiently large makes the last probability
at most $n^{-3}$ in either case.  Finally, since $A$ is a fixed constant,
\[
  \frac1a
  =
  \begin{cases}
    \Omega\!\left(\min\left\{m_0^{-1},\sqrt{n/\log n}\right\}\right),
      & \text{for 3-Majority},\\[1mm]
    \Omega\!\left(\min\left\{m_0^{-1},n/\log n\right\}\right),
      & \text{for 2-Choices}.
  \end{cases}
\]
This proves the theorem.
\end{proof}


\section{Analysis: Plurality Consensus}
In this section we prove \Cref{thm:plurality-informal}. In the next lemma, we show that the bias between the plurality opinion and any other opinion is preserved until the balancing time $\tau_{\mathrm b}$. For 2-Choices, we show that the scaled bias $\eta_t(i,j)$ is preserved until the balancing time $\tau_{\mathrm b}$, which is a stronger condition.
\begin{lemma}[Plurality bias up to the balancing time]
\label{lem:plurality-before-balance}
Assume that
\[
  \alpha_0(1)>\alpha_0(2)\ge\cdots\ge\alpha_0(k),
  \qquad
  m_0=\alpha_0(1),
  \qquad
  \bias_0=\alpha_0(1)-\alpha_0(2).
\]
Under the assumptions of \Cref{thm:plurality-informal}, with probability at
least $1-n^{-4}$, at the time $s=\tau_{\mathrm b}$ given by
\Cref{lem:quick-balance}, opinion $1$ is still the plurality.  Moreover, for
every $j\ne 1$, in 3-Majority,
\[
  \delta_s(1,j)\ge c\bias_0
\]
and, in 2-Choices,
\[
  \frac{\delta_s(1,j)^2}{m_s}
  \ge
  c\frac{\bias_0^2}{m_0},
\]
for a universal constant $c>0$.
\end{lemma}

\begin{proof}
Write $B=\bias_0$.  Let $C_{\mathrm b}$ and $c_{\mathrm b}$ be the
constants from \Cref{lem:quick-balance}, put
\[
  S=\left\lceil C_{\mathrm b}m_0^{-1}\right\rceil,
  \qquad
  \rho=c/2,
\]
and let $\mathcal E_{\mathrm b}$ be the successful event supplied by
\Cref{lem:quick-balance}.  Explicitly, on this event,
\[
  \tau_{\mathrm b}\le S,
  \qquad
  m_t\ge c \, m_0
  \quad\text{for every }0\le t\le\tau_{\mathrm b}.
\]
Moreover, $\Pr[\mathcal E_{\mathrm b}^c]\le n^{-5}$.
The deterministic cutoff $S$ is retained in the stopping times below because
$\tau_{\mathrm b}\le S$ holds only on the high-probability event
$\mathcal E_{\mathrm b}$, not unconditionally.  It gives the martingale bounds
an unconditional finite horizon; after applying the bounds, we intersect the
resulting event with $\mathcal E_{\mathrm b}$.

\textbf{3-Majority.} Define the stopping times
\[
  \tau_{\Delta}
  =
  \inf\{t\ge0:\exists \ell\ne1
  \text{ such that }\delta_t(1,\ell)\le B/2\},
  \qquad
  \tau_m=\inf\{t\ge0:\alpha_t(1)\le\rho m_0\}.
\]
We prove that
\begin{equation}
\label{eq:delta-loss-before-balance}
  \Pr[\tau_{\Delta}\le\tau_{\mathrm b}]\le n^{-4}.
\end{equation}
The 3-Majority part follows immediately from \Cref{eq:delta-loss-before-balance},
since if
$s=\tau_{\mathrm b}<\tau_{\Delta}$, then
$\delta_s(1,j)>B/2$ for every $j\ne1$.

Fix $j\ne1$ and set
\[
  \sigma=\min\{\tau_{\mathrm b},\tau_{\Delta},\tau_m,S\}.
\]
For $t<\sigma$, opinion $1$ is the unique plurality, and hence
$\alpha_t(1)=m_t$.  Moreover, $t<\tau_{\mathrm b}$ implies
$\gamma_t<m_t/2$, and so
\[
  a_t:=\alpha_t(1)+\alpha_t(j)-\gamma_t
  \ge \alpha_t(1)-\gamma_t
  > \alpha_t(1)/2.
\]
For $t\ge\sigma$ set $a_t=0$.  Let $G_0=1$ and
$G_{t+1}=G_t(1+a_t)$.  Finally define
\[
  Z_t=\frac{\delta_{t\wedge\sigma}(1,j)}{G_t}.
\]
By \Cref{item:delta} of \Cref{lem:basic inequality}, for $t<\sigma$,
\[
  \E_t[\delta_{t+1}(1,j)]
  =
  \delta_t(1,j)(1+a_t),
\]
while for $t\ge\sigma$ both the numerator and $G_t$ are stopped.  Hence
$(Z_t)_{t\ge0}$ is a martingale.

We now bound the lower tail of this martingale.  For $t<\sigma$,
\[
  Z_t-Z_{t+1}
  =
  \frac{\E_t[\delta_{t+1}(1,j)]-\delta_{t+1}(1,j)}
       {G_{t+1}} .
\]
By \Cref{item:BC for delta} of
\Cref{lem:Bernstein condition for sync processes},
$\E_t[\delta_{t+1}(1,j)]-\delta_{t+1}(1,j)$ satisfies a one-sided Bernstein
condition with parameters
\[
  \frac{2}{n}
  \qquad\text{and}\qquad
  \frac{2(\alpha_t(1)+\alpha_t(j))}{n}.
\]
Since $G_{t+1}$ is $\calF_t$-measurable, scaling by $1/G_{t+1}$ preserves the
one-sided Bernstein condition with both parameters scaled in the standard way.
Thus the increment $Z_t-Z_{t+1}$ satisfies a one-sided Bernstein condition with
boundedness parameter at most $2/n$ and predictable variance proxy
\[
  v_t
  =
  \indicator_{\{t<\sigma\}}
  \frac{2(\alpha_t(1)+\alpha_t(j))}{nG_{t+1}^2}
  \le
  \frac{8a_t}{nG_{t+1}^2},
\]
where the last inequality uses $a_t>m_t/2$ and
$\alpha_t(1)+\alpha_t(j)\le2m_t$ on $\{t<\sigma\}$, while both sides are zero
on $\{t\ge\sigma\}$.  The variance proxies telescope:
\[
  \sum_{t<S}v_t
  \le
  \frac{8}{n}\sum_{t<S}\frac{a_t}{G_{t+1}^2}
  \le
  \frac{8}{n}\sum_{t<S}
  \left(\frac1{G_t}-\frac1{G_{t+1}}\right)
  \le \frac{8}{n}.
\]
Applying \Cref{cor:Freedman-predictable} to the martingale
$X_t=Z_0-Z_t$ with $h=B/2$, $\bounded=2/n$, and $v=8/n$ gives
\begin{equation}
\label{eq:delta-normalized-freedman}
  \Pr\!\left[\exists t\le S:\ Z_0-Z_t\ge B/2\right]
  \le
  \exp(-\Omega(nB^2)).
\end{equation}

If $\tau_{\Delta}\le\min\{\tau_{\mathrm b},\tau_m,S\}$ and opinion $j$ is one
of the opinions attaining the minimum in the definition of $\tau_{\Delta}$,
then, at time $\sigma=\tau_{\Delta}$,
\[
  Z_0-Z_{\sigma}
  =
  \delta_0(1,j)-\frac{\delta_{\sigma}(1,j)}{G_{\sigma}}
  \ge
  B-\frac{B}{2}
  =
  B/2,
\]
because $\delta_0(1,j)\ge B$ and $G_{\sigma}\ge1$.  Taking a union bound over
the at most $n-1$ competing opinions and using
\Cref{eq:delta-normalized-freedman}, we get
\[
  \Pr[\tau_{\Delta}\le\min\{\tau_{\mathrm b},\tau_m,S\}]
  \le n\exp(-\Omega(nB^2)).
\]

On $\mathcal E_{\mathrm b}$, if $t<\tau_{\Delta}\wedge\tau_{\mathrm b}$ then
opinion $1$ is the plurality, and therefore
\[
  \alpha_t(1)=m_t\ge c \, m_0>\rho m_0.
\]
Thus $\tau_m$ cannot occur before
$\tau_{\Delta}\wedge\tau_{\mathrm b}$ on $\mathcal E_{\mathrm b}$, and
$\tau_{\mathrm b}\le S$ there.  Consequently,
\[
  \Pr[\tau_{\Delta}\le\tau_{\mathrm b}]
  \le
  \Pr[\mathcal E_{\mathrm b}^c]
  +
  \Pr[\tau_{\Delta}\le\min\{\tau_{\mathrm b},\tau_m,S\}]
  \le
  n^{-5}+n\exp(-\Omega(nB^2)).
\]
Since the 3-Majority plurality assumption gives
$B=\omega(\sqrt{\log n/n})$, the last display is at most $n^{-4}$ for all
sufficiently large $n$.  This proves the 3-Majority part.

\textbf{2-Choices.} Put
\[
  q=\frac{B}{2\sqrt{m_0}}
\]
and define the global scaled-bias loss time
\[
  \tau_{\eta}
  =
  \inf\{t\ge0:\exists \ell\ne1
  \text{ such that }\eta_t(1,\ell)\le q\}.
\]
Let
\[
  \sigma_0=\min\{\tau_{\mathrm b},\tau_{\eta},S\}.
\]
For $t<\sigma_0$, all scaled biases $\eta_t(1,\ell)$ are positive, and hence
opinion $1$ is the unique plurality: $\alpha_t(1)=m_t$.  It also holds that
$\gamma_t<\alpha_t(1)/2$.  Define
\[
  b_t=
  \begin{cases}
    \alpha_t(1)-\gamma_t, & t<\sigma_0,\\
    0, & t\ge\sigma_0,
  \end{cases}
  \qquad
  H_0=1,\qquad H_{t+1}=H_t(1+b_t).
\]
Then
\[
  A_t=\frac{\alpha_{t\wedge\sigma_0}(1)}{H_t}
\]
is a martingale by \Cref{item:alpha} of \Cref{lem:basic inequality}.  Let
\[
  \tau_A=\inf\{t\ge0:\ A_t\ge2m_0\}.
\]
We first show that this stopping time is unlikely before
the scaled bias is lost.  Let $\sigma_A=\min\{\sigma_0,\tau_A\}$ and stop
$A_t$ at $\sigma_A$.  For $t<\sigma_A$,
\[
  A_t<2m_0,\qquad
  \gamma_t<\alpha_t(1)/2,
\]
and so the 2-Choices variance proxy from \Cref{item:BC for alpha} of
\Cref{lem:Bernstein condition for sync processes}, after scaling by the
$\calF_t$-measurable factor $1/H_{t+1}$, is at most
\[
  u_t
  =
  \indicator_{\{t<\sigma_A\}}
  \frac{\alpha_t(1)(\alpha_t(1)+\gamma_t)}{nH_{t+1}^2}
  \le
  \indicator_{\{t<\sigma_A\}}\frac{6m_0^2}{n}.
\]
Since $S=O(m_0^{-1})$, for a universal constant $C_A$,
\[
  \sum_{t<S}u_t\le C_A\frac{m_0}{n}.
\]
Moreover, the corresponding boundedness parameter is at most $1/n$.  Notice
that $A_0=m_0$.  Hence, on the event $\{\tau_A\le\sigma_0\}$, we have
$\sigma_A=\tau_A\le S$ and
\[
  A_{\tau_A\wedge\sigma_A}-A_0
  =
  A_{\tau_A}-m_0
  \ge m_0,
\]
because $A_{\tau_A}\ge2m_0$ by definition of $\tau_A$.  Therefore
\[
  \{\tau_A\le\sigma_0\}
  \subseteq
  \left\{\exists t\le S:\ A_{t\wedge\sigma_A}-A_0\ge m_0\right\}.
\]
Applying \Cref{cor:Freedman-predictable} to the stopped martingale
$A_{t\wedge\sigma_A}-A_0$ with $h=m_0$ gives
\begin{equation}
\label{eq:alpha-normalized-upper-tail}
  \Pr[\tau_A\le\sigma_0]
  \le
  \exp(-\Omega(nm_0)).
\end{equation}

Now fix $j\ne1$ and define
\[
  \sigma=\min\{\tau_{\mathrm b},\tau_{\eta},\tau_A,S\}.
\]
Set
\[
  a_t=
  \begin{cases}
    \alpha_t(1)+\alpha_t(j)-\gamma_t, & t<\sigma,\\
    0, & t\ge\sigma,
  \end{cases}
  \qquad
  G_0=1,\qquad G_{t+1}=G_t(1+a_t).
\]
Since $a_t=b_t+\alpha_t(j)\ge b_t$ for $t<\sigma$, we have $G_t\ge H_t$ for
every $t\le\sigma$.  Also,
\[
  D_t=\frac{\delta_{t\wedge\sigma}(1,j)}{G_t}
\]
is a martingale, again by \Cref{item:delta} of
\Cref{lem:basic inequality}.

For $t<\sigma$, opinion $1$ is the plurality and the configuration is
unbalanced, so
\[
  \alpha_t(1)=m_t,\qquad
  \gamma_t<\alpha_t(1)/2,\qquad
  \alpha_t(j)\le\alpha_t(1).
\]
Since $t<\tau_A$, we also have
\begin{equation}
\label{eq:alpha-bound-by-normalizers}
  \alpha_t(1)=A_tH_t<2m_0H_t\le2m_0G_t .
\end{equation}
The 2-Choices variance proxy from \Cref{item:BC for delta}, after scaling by
$1/G_{t+1}$, is
\[
  w_t
  =
  \indicator_{\{t<\sigma\}}
  \frac{(\alpha_t(1)+\alpha_t(j))
  (\alpha_t(1)+\alpha_t(j)+\gamma_t)}{nG_{t+1}^2}.
\]
Using the preceding bounds and \Cref{eq:alpha-bound-by-normalizers},
\[
  w_t
  \le
  \indicator_{\{t<\sigma\}}\frac{20m_0^2}{n}.
\]
Since $S=O(m_0^{-1})$, we obtain that, for a universal constant $C_D$,
\[
  \sum_{t<S}w_t\le C_D\frac{m_0}{n}.
\]
The boundedness parameter is at most $2/n$.  Let
$c_\eta=1-1/\sqrt2$.  Applying \Cref{cor:Freedman-predictable} to the
martingale $D_0-D_t$, with $h=c_\eta B$ and $v=C_Dm_0/n$, yields
\begin{equation}
\label{eq:eta-normalized-freedman}
  \Pr\!\left[\exists t\le S:\ D_0-D_t\ge c_\eta B\right]
  \le
  \exp\!\left(-\Omega\!\left(\frac{nB^2}{m_0}\right)\right),
\end{equation}
where we used $B\le m_0$.

If $\tau_\eta<\min\{\tau_{\mathrm b},\tau_A,S\}$ and opinion $j$ realises the
minimum in the definition of $\tau_\eta$, then either
$\delta_{\tau_\eta}(1,j)<0$, in which case $D_{\tau_\eta}<0$, or
$\delta_{\tau_\eta}(1,j)\ge0$, in which case
\[
  \delta_{\tau_\eta}(1,j)
  \le
  q\sqrt{\alpha_{\tau_\eta}(1)}
  \le
  \frac{B}{2\sqrt{m_0}}\sqrt{2m_0H_{\tau_\eta}}
  =
  \frac{B}{\sqrt2}\sqrt{H_{\tau_\eta}}.
\]
Since $G_{\tau_\eta}\ge H_{\tau_\eta}\ge1$, in both cases
\[
  D_{\tau_\eta}
  =
  \frac{\delta_{\tau_\eta}(1,j)}{G_{\tau_\eta}}
  \le
  \frac{B}{\sqrt2}.
\]
As $D_0=\delta_0(1,j)\ge B$, the event
$\{\tau_\eta<\min\{\tau_{\mathrm b},\tau_A,S\}\}$ forces
$D_0-D_{\tau_\eta}\ge c_\eta B$ for one of the at most $n-1$ choices of $j$.
By \Cref{eq:eta-normalized-freedman} and a union bound,
\begin{equation}
\label{eq:eta-loss-before-stops}
  \Pr[\tau_\eta<\min\{\tau_{\mathrm b},\tau_A,S\}]
  \le
  n\exp\!\left(-\Omega\!\left(\frac{nB^2}{m_0}\right)\right).
\end{equation}

On $\mathcal E_{\mathrm b}$, $\tau_{\mathrm b}\le S$.  Combining this with
\Cref{eq:alpha-normalized-upper-tail,eq:eta-loss-before-stops},
\[
  \Pr[\tau_\eta\le\tau_{\mathrm b}]
  \le
  n^{-5}
  +\exp(-\Omega(nm_0))
  +n\exp\!\left(-\Omega\!\left(\frac{nB^2}{m_0}\right)\right).
\]
The 2-Choices assumptions in \Cref{thm:plurality-informal} give
$nm_0=\omega((\log n)^2)$ and $nB^2/m_0=\omega(\log n)$, so the last display
is at most $n^{-4}$ for all sufficiently large $n$.  Thus, with probability at
least $1-n^{-4}$, $s=\tau_{\mathrm b}<\tau_\eta$.  At that time, for every
$j\ne1$,
\[
  \eta_s(1,j)>\frac{B}{2\sqrt{m_0}},
\]
opinion $1$ is still the plurality, and therefore
\[
  \frac{\delta_s(1,j)^2}{m_s}
  =
  \eta_s(1,j)^2
  \ge
  \frac{B^2}{4m_0}.
\]
This proves the 2-Choices part, and the lemma follows with a universal
constant $c>0$.
\end{proof}
Now we can complete the proof of \Cref{thm:plurality-informal} by applying the
plurality-consensus result in Theorem~2.6 of Shimizu and
Shiraga~\cite{Shimizu2025}
from the random time $s=\tau_{\mathrm b}$.
Conditioned on the current configuration, \Cref{lem:plurality-before-balance} ensures that its conditions are satisfied.
\begin{proof}[Proof of \Cref{thm:plurality-informal}]
At the time $s=\tau_{\mathrm b}$ above, the configuration is balanced and
$m_s=\Omega(m_0)$.  In 3-Majority,
\[
  \delta_s(1,j)\ge c\bias_0
  =
  \omega\!\left(\sqrt{\frac{\log n}{n}}\right)
  \qquad\text{for every }j\ne 1.
\]
In 2-Choices, the normalized bound from
\Cref{lem:plurality-before-balance} gives
\[
  \delta_s(1,j)
  \ge
  \omega\!\left(\sqrt{\frac{m_s\log n}{n}}\right)
  \qquad\text{for every }j\ne 1.
\]
Hence Theorem~2.6 of Shimizu and Shiraga can be applied
from the random balanced configuration at time $s$, conditioned on the current
state.  We use the standard high-probability form of their theorem with the
constants in the assumptions chosen so that the conditional failure
probability is at most $n^{-4}$.  It gives consensus on opinion $1$ within
$O(m_0^{-1}\log n)$ additional rounds.  Combining this with the $n^{-4}$
failure probability in \Cref{lem:plurality-before-balance} and with
\Cref{lem:quick-balance}, and then increasing constants if necessary, the
total failure probability is at most $n^{-2}$.  Since \Cref{lem:quick-balance}
gives $s=O(m_0^{-1})$, this proves \Cref{thm:plurality-informal}.
\end{proof}
\section{Conclusion and future directions}
\label{sec:conclusion}

We identified the maximum initial opinion density
\(m_0=\|\alphavec^{(0)}\|_\infty\) as the parameter governing the
consensus time of both \threemaj and \twochoices on the complete graph.  In
particular, \threemaj reaches consensus in
\(\tilde{\Theta}(\min\{m_0^{-1},\sqrt n\})\) rounds, while \twochoices reaches
consensus in \(\tilde{\Theta}(m_0^{-1})\) rounds.  Thus the convergence time is
controlled by the extremal statistic $m_0$, rather than by the number of
opinions or an aggregate statistic such as the squared \(\ell_2\)-norm
\(\gamma_0\).

This gives substantial improvements over the bounds of Shimizu and
Shiraga~\cite{Shimizu2025} for unbalanced configurations, where \(m_0\) can be
much larger than \(\gamma_0\).  For \twochoices, the improvement can be as large
as a factor \(\sqrt n\); for \threemaj, it can be as large as a factor
\(n^{1/4}\).

It would be natural to investigate whether the same infinity-norm
parameterisation extends to asynchronous \threemaj and \twochoices.  Since a
single asynchronous activation has the same local transition probabilities as
one vertex update in the synchronous process, one possible route is to adapt the
existing \(\ell_2\)-norm analysis in the asynchronous setting.  If such an
analogue is available, our reduction from unbalanced to balanced configurations
should yield the corresponding asynchronous bounds up to the expected factor
\(n\) in the time scale.

This work also exposes configurations in which \threemaj is slower than the
Undecided-State, giving a first explicit gap between these two dynamics.  More
precisely, \threemaj can be slower by a factor \(m_0/\gamma_0\), which becomes
substantial for unbalanced configurations.  It would be interesting to build on
the analyses of~\cite{Becchetti2016,cooper2026undecidedstatedynamicsopinions}
to determine whether the consensus time of the \undecided is
\(\tilde{\Theta}(\gamma_0/m_0^2)\) for arbitrary numbers of opinions, which
would extend this gap beyond the currently known range.

\section{Acknowledgements}
The author thanks Francesco d'Amore and Emanuele Natale for helpful discussions and feedback on this work.
This work has been supported by the AID INRIA-DGA project n°2023000872 ``BioSwarm''.


\bibliography{references}
\appendix

\section{Beyond constant samples and other settings}
\label{sec:beyond_constant_samples}
The \hmaj with more than a constant number of samples is much less understood than the \threemaj. Becchetti et al.~\cite{becchettiSimpleDynamicsPlurality2017} proved a lower bound of \(\Omega(k/h^2)\) rounds to reach consensus (that holds w.h.p.), which in particular implies that increasing the sample size $h$ to be polylogarithmic in $n$ is not sufficient for efficient consensus for a large number of opinions.
\cite{damore_et_al:LIPIcs.DISC.2025.27} studied the \hmaj dynamics for \(h =\Omega(k \log n)\), assuming an initial bias of $\omega(\sqrt{\frac{m_0 \log n}{n}})$ towards the plurality opinion, while the analysis of the \hmaj dynamics is still open for general $k$ and $h$.
A recent work \cite{damore2026dejavuminimalisticmechanismdistributed} studied a different simple dynamics, the \textit{deja-vu} dynamics, in which an agent samples opinions until a repetition occurs and adopts the repeated opinion. If no repetition occurs, it keeps its current opinion. It can be viewed as an adaptive-sampling generalisation of \twochoices. They show that, assuming an initial bias of $\omega(\sqrt{\frac{m_0 \log n}{n}})$ towards the plurality opinion, the \textit{deja-vu} dynamics reaches consensus in $O(\sqrt{m_0^{-1}} \log n)$ rounds w.h.p., regardless of the number of opinions. 

The \threemaj, the \hmaj, and other majority-based dynamics have also been investigated in other settings.
For example, this includes the presence of communication noise (which corrupts the exchanged messages), the presence of stubborn agents \cite{DAmoreZ22,dAmoreZ25}, and settings in which some opinion is preferred over the others, that is, there is some probability that an agent, instead of running the protocol, spontaneously switches to the preferred opinion \cite{CrucianiMQR23,LesfariGP22}.

More generally, work on opinion dynamics fits into the recent trend in the distributed computing community of drawing inspiration from social and biological systems.
Other than the consensus problem, researchers have analysed distributed search problems as well as more generic problems (MIS, leader election, etc.) in extremely weak computational models (such as the beeping model or the stone-age model) \cite{KormanV23,FeinermanK17,ClementiDGN21,GiakkoupisZ23,GiakkoupisTZ24,DArchivioV24,FuggerNR24,vacus-ziccardi2025,FraigniaudKR16,FeinermanHK14,FraigniaudN19}.

\section{Tools}
We shall use the following results concerning (super)martingales.
\begin{theorem}[Freedman's inequality; \cite{Freedman1975}]\label{thm:Freedman}
  Let $(M_t)_{t\ge0}$ be a martingale with $M_0=0$ and increments
  $\Delta_t=M_t-M_{t-1}\le B$ almost surely.  Let
  $V_t=\sum_{s=1}^t\E[\Delta_s^2\mid\mathcal{F}_{s-1}]$.  Then, for every
  $\lambda,v>0$,
  \[
    \Pr\!\left[\exists t:\ M_t\ge \lambda\ \text{and}\ V_t\le v\right]
    \le \exp\!\left(-\frac{\lambda^2}{2(v+B\lambda/3)}\right).
  \]
\end{theorem}

\begin{corollary}[Corollary 3.8 in \cite{Shimizu2025}] 
      \label{cor:Freedman}
        Let $(X_t)_{t\in \Nat_0}$ be a supermartingale associated with the natural filtration $(\calF_t)_{t\in \Nat_0}$.
        Suppose that, for every $t\ge 1$, the difference $X_t - X_{t-1}$ conditioned on $\calF_{t-1}$ satisfies one-sided $(\bounded,\variance)$-Bernstein condition.
        Then, for any $h > 0$, we have
      \begin{align*}
         \Pr\sbra*{{}^{\exists} t\le T,\,X_t-X_0\geq h} \leq \exp\rbra*{-\frac{h^2/2}{T\variance + (h\bounded )/3}}.
      \end{align*}
    \end{corollary}

\begin{theorem}[Ville's inequality; \cite{Ville1939}]
\label{thm:Ville}
Let $(L_t)_{t\ge0}$ be a nonnegative supermartingale.  Then, for every
$T\in\Nat_0$ and every $a>0$,
\[
  \Pr\!\left[\exists t\le T:\ L_t\ge a\right]
  \le \frac{\E[L_0]}{a}.
\]
\end{theorem}

\begin{corollary}[Predictable-variance Freedman inequality under one-sided Bernstein condition]
  \label{cor:Freedman-predictable}
  Let $(X_t)_{t\in\Nat_0}$ be a supermartingale associated with the natural
  filtration $(\calF_t)_{t\in\Nat_0}$.  Suppose that, for every $t\ge1$, the
  difference $X_t-X_{t-1}$ conditioned on $\calF_{t-1}$ satisfies a one-sided
  $(\bounded_t,\variance_t)$-Bernstein condition, where
  $\bounded_t,\variance_t$ are $\calF_{t-1}$-measurable,
  $\bounded_t\le\bounded$, and
  \[
    \sum_{t=1}^T \variance_t\le v
  \]
  almost surely.  Then, for any $h>0$,
  \[
    \Pr\!\left[\exists t\le T:\ X_t-X_0\ge h\right]
    \le
    \exp\!\left(-\frac{h^2/2}{v+h\bounded/3}\right).
  \]
\end{corollary}

\begin{proof}
Fix \(0<\lambda<3/\bounded\) and put
\[
  c_\lambda=\frac{\lambda^2/2}{1-\lambda\bounded/3}.
\]
By \(\bounded_t\le \bounded\) and the one-sided Bernstein condition,
\[
  \E\!\left[
    \exp(\lambda(X_t-X_{t-1}))\mid\calF_{t-1}
  \right]
  \le
  \exp(c_\lambda \variance_t).
\]
Hence
\[
  L_t=
  \exp\!\left(
    \lambda(X_t-X_0)-c_\lambda\sum_{s=1}^t\variance_s
  \right)
\]
is a nonnegative supermartingale. By \Cref{thm:Ville},
\[
  \Pr\!\left[\exists t\le T:\ X_t-X_0\ge h\right]
  \le
  \exp(-\lambda h+c_\lambda v),
\]
using \(\sum_{s=1}^T\variance_s\le v\) almost surely. Choosing
\[
  \lambda=\frac{h}{v+h\bounded/3}
\]
gives
\[
  \exp\!\left(-\frac{h^2/2}{v+h\bounded/3}\right).
\]
\end{proof}

\begin{lemma}[Lemma 4.1 in \cite{Shimizu2025}]\label{lem:basic inequality}
  Consider the quantities defined in \Cref{sec:model} for 3-Majority or 2-Choices.
  Then, we have the following for any $t\geq 1$:
\begin{enumerate}
            \renewcommand{\labelenumi}{(\roman{enumi})}
\item \label{item:alpha} 
For any opinion $i \in [k]$, we have
  \begin{align*} 
     &\E_{t-1} [\alpha_t(i)] = 
      \alpha_{t-1}(i) \qty( 1 + \alpha_{t-1}(i) - \alphanorm_{t-1} ), \\ 
    &\Var_{t-1}[\alpha_t(i)] \le \begin{cases}
      \frac{\alpha_{t-1}(i)}{n}	& \text{for 3-Majority},\\
      \frac{\alpha_{t-1}(i)(\alpha_{t-1}(i) + \alphanorm_{t-1})}{n}	& \text{for 2-Choices}.
    \end{cases} 
  \end{align*}
\item \label{item:delta} 
For any two distinct opinions $i,j\in [k]$, we have
  \begin{align*} 
     &\E_{t-1} [\delta_t(i,j)] =
      \delta_{t-1}(i,j) \qty( 1 + \alpha_{t-1}(i) + \alpha_{t-1}(j) - \alphanorm_{t-1} ), \\
    &\Var_{t-1}[\delta_t(i,j)] \le \begin{cases}
      \frac{2}{n}(\alpha_{t-1}(i) + \alpha_{t-1}(j))	& \text{for 3-Majority},\\[10pt]
      \frac{1}{n}(\alpha_{t-1}(i)+\alpha_{t-1}(j))(\alpha_{t-1}(i)+\alpha_{t-1}(j)+\alphanorm_{t-1})	& \text{for 2-Choices}.
    \end{cases} 
  \end{align*}
\end{enumerate}
\end{lemma}

\begin{lemma}[Lemma 4.2 in \cite{Shimizu2025}]
\label{lem:Bernstein condition for sync processes}
Consider the quantities defined in \Cref{sec:model} for 3-Majority or 2-Choices.
Then, we have the following for any $t\geq 1$:
\begin{enumerate}
            \renewcommand{\labelenumi}{(\roman{enumi})}
    \item \label{item:BC for alpha}
    For any opinion $i\in [k]$, both
    $\alpha_{t}(i)-\E_{t-1}\sbra*{\alpha_t(i)}$ and
    $\E_{t-1}\sbra*{\alpha_t(i)}-\alpha_t(i)$ conditioned on round $t-1$
    satisfy the one-sided $\qty(\frac{1}{n},\variance)$-Bernstein condition,
    where
    \begin{align*} 
       \variance = \begin{cases}
        \frac{\alpha_{t-1}(i)}{n}	& \text{for 3-Majority},\\
        \frac{\alpha_{t-1}(i)(\alpha_{t-1}(i)+\alphanorm_{t-1})}{n} & \text{for 2-Choices}.
       \end{cases}
    \end{align*}
    \item \label{item:BC for delta}
     For any two distinct opinions $i,j\in [k]$, both
     $\delta_{t}(i,j)-\E_{t-1}\sbra*{\delta_t(i,j)}$ and
     $\E_{t-1}\sbra*{\delta_t(i,j)}-\delta_{t}(i,j)$, conditioned on round
     $t-1$, satisfy the one-sided $ \qty(\frac{2}{n}, \variance) $-Bernstein
     condition, 
     where
    \begin{align*} 
       \variance = \begin{cases}
          \frac{2}{n}(\alpha_{t-1}(i)+\alpha_{t-1}(j))	& \text{for 3-Majority},\\[10pt]
          \frac{1}{n}\qty(\alpha_{t-1}(i) + \alpha_{t-1}(j))\qty(\alpha_{t-1}(i)+\alpha_{t-1}(j)+\alphanorm_{t-1}) & \text{for 2-Choices}.
       \end{cases} 
    \end{align*}
\end{enumerate}
\end{lemma}

\end{document}